\newif\ifEditMode
\EditModefalse

\documentclass{eptcs}
\providecommand{\event}{GandALF 2020} 
\usepackage{underscore}           

\usepackage{graphicx, amssymb, tikz, mathtools, tikz-cd, epstopdf, verbatim}
\usepackage{amsbsy,scalerel, xspace, xcolor, paralist, amsmath, amsthm}
\usepackage{enumitem}
\PassOptionsToPackage{bookmarks=false}{hyperref}

\newtheorem{theorem}{Theorem}[section]
\newtheorem{note}{Note}[section]

\newtheorem{definition}[theorem]{Definition}
\newtheorem{corollary}[theorem]{Corollary}
\newtheorem{lemma}[theorem]{Lemma}
\newtheorem{proposition}[theorem]{Proposition}

\graphicspath{{figures/}}

\allowdisplaybreaks

\newcounter{assumption}[section]

\tikzcdset{
arrow style=tikz,
diagrams={>={Straight Barb[scale=0.8]}}
}

\renewcommand{\setminus}[0]{-}

\newcommand{\setunion}[0]{\boldsymbol{\cup}}
\newcommand{\setint}[0]{\boldsymbol{\cap}}

\newcommand{\setArg}[2]{\left\{ {#1} \,\,\middle|\,\, {#2} \right\}}

\newcommand{\card}[1]{\left| {#1} \right|}
\newcommand{\transp}{\top}

\newcommand{\shared}[0]{\textsf{shared}}
\newcommand{\illegal}[0]{\textsf{Illegal}}
\newcommand{\comp}[0]{\textsf{Cmp}}
\newcommand{\exten}[0]{\textsf{ExtEn}}
\newcommand{\extdest}[0]{\textsf{ExtDest}}

\newcommand{\epsclosure}[0]{\textsf{IntReach}}
\newcommand{\mirror}[0]{mirror\xspace}

\newcommand{\id}[0]{\text{id}}
\newcommand{\flip}[0]{\text{flip}}
\newcommand{\structName}[0]{preordered heap\xspace}
\newcommand{\StructName}[0]{Preordered heap\xspace}
\newcommand{\compStructName}[0]{sieved heap\xspace}
\newcommand{\CompStructName}[0]{Sieved heap\xspace}
\newcommand{\propName}[0]{admissibility\xspace}
\newcommand{\sgroup}[0]{S}

\newcommand{\inv}[0]{\gamma}
\newcommand{\smult}[0]{\mu}
\newcommand{\tmult}[0]{\tau}
\newcommand{\smultL}[1]{\smult_{#1}}
\newcommand{\smultR}[1]{\smult^{#1}}
\newcommand{\sadj}[0]{\hat \smult}
\newcommand{\sadjL}[1]{\left(\smult_{#1}\right)^{*}}
\newcommand{\sadjR}[1]{\left(\smult^{#1}\right)^{*}}

\newcommand{\tmultL}[1]{\tmult_{#1}}
\newcommand{\tmultR}[1]{\tmult^{#1}}
\newcommand{\tadj}[0]{\hat \tmult}
\newcommand{\tadjL}[1]{\left(\tmult_{#1}\right)^{*}}
\newcommand{\tadjR}[1]{\left(\tmult^{#1}\right)^{*}}

\renewcommand{\paragraph}[1]{\noindent\textbf{#1}}

\renewcommand{\tt}{\texttt}

\title{The Quotient in Preorder Theories}

\def\titlerunning{The Quotient in Preorder Theories} 

\author{
{\'I}{\~n}igo X. {\'I}ncer Romeo
\institute{University of California, Berkeley, USA}
\email{inigo@eecs.berkeley.edu}
\and
Leonardo Mangeruca
\institute{Raytheon Technologies Research Center,  Rome, Italy}
\email{leonardo.mangeruca@rtx.com}
\and
Tiziano Villa
\institute{Universit{\` a} di Verona, Italy}
\email{tiziano.villa@univr.it}
\and
Alberto Sangiovanni-Vincentelli
\institute{University of California, Berkeley, USA}
\email{alberto@eecs.berkeley.edu}
}
\def\authorrunning{{\'I}. {\'I}ncer Romeo, L. Mangeruca, T. Villa, and A. Sangiovanni-Vincentelli} 

\begin{document}
\maketitle

\begin{abstract}
Seeking the largest solution to an expression of the form $A x \le B$ is a common task in
several domains of engineering and computer science. This largest solution is
commonly called quotient. Across domains, the meanings of the binary
operation and the preorder are quite different, yet the syntax for computing the largest
solution is remarkably similar. This paper is about finding a
common framework to reason about quotients. 
We only assume we operate on a
preorder endowed with an abstract monotonic multiplication and an involution. We
provide a condition, called \propName, which guarantees the existence of the quotient, and which yields
its closed form. We call \structName{s} those
structures satisfying the \propName condition. We show that many existing
theories in computer science are \structName{s}, and we are thus able to
derive a quotient for them, subsuming existing solutions when available
in the literature. We introduce the concept of \compStructName{s} to deal with
structures which are given over multiple domains of definition. We show that
\compStructName{s} also have well-defined quotients.
\end{abstract}

\section{Introduction}

The identification of missing objects is a common task in engineering. Suppose
an engineer wishes to implement a design with a mathematical description $B$, and will use a component with a description $A$ to implement this design. In order to find out what needs to be added to $A$ in order to implement $B$, the engineer seeks a component $x$ in an expression of the form $A \bullet x = B$, where $\bullet$ is an operator yielding the composite of two design elements. Many compositional theories include the notion of a preorder, usually called refinement. The statement $A \le C$ usually reads ``$A$ refines $C$'' or ``$A$ is more specific than $C$.'' In this setting, the problem is recast as finding an $x$ such that $A \bullet x \le B$. It is often assumed that the composition operation is monotonic with respect to this preorder. Therefore, if $x$ is a solution, so is any $y$ satisfying $y \le x$. This focuses our attention on finding the largest $x$ that satisfies the expression. The literature often calls this largest solution \emph{quotient}. 

\subsection{Background}
The logic synthesis community has been a pioneer in defining and 
solving special cases of the quotient problem for combinational and sequential
logic circuit design (\cite{kim72,cerny-marin77}) under names like circuit 
rectification or engineering change or component replacement. 
In combinational synthesis, much work has been reported to support algebraic 
and Boolean division: given dividend $f$ and divisor $g$, find the quotient 
$q$ and remainder $r$ such $f = q\cdot g + r$ (for $\cdot ,+$ standard Boolean operators AND and OR, respectively), 
as key operation to restructure multi-level Boolean networks~\cite{iwls10:ch2}.
The quotient problem for combinational circuits was formulated as a general
replacement problem in~\cite{burch-comb-cycle}: given the 
combinational circuits $A$ and $C$ whose synchronous composition produces 
the circuit specification $B$, what are the legal replacements of $C$ that
are consistent with the input-output relation of $B$? The valid replacements
for $C$ were defined as the combinational circuits $x$ such that 
$A \circ x \subseteq B$, and the largest solution for $x$ was characterized by
the closed formula $x = \left(A \circ B^{\bot}\right)^{\bot}$,
where $(\cdot)^\bot$ is a unary operator that complements the input-output relation 
of the circuit to which it is applied (switching the inputs and outputs),
while a hiding operation gets rid of the internal signals.

In sequential optimization, the typical question addressed was, 
given a finite-state machine (FSM) $A$, find an FSM $x$
such that their synchronous composition produces an FSM behaviorally equivalent 
to a specification FSM $B$, i.e., solve over FSMs the equation 
$A \circ x = B$, where $\circ$ is synchronous composition and equality
is FSM input-output equivalence.
Various topologies were solved, starting with serial composition where 
the unknown was either the head or tail machine, to more complex 
interconnections with feedback.
As a matter of fact, sometimes both $A$ and $x$ were known, but the goal
was to change them into FSMs yielding better logical implementations, while
preserving their composition, with the objective to optimize a sequential 
circuit by computing and exploiting the flexibility due to its modular 
structure and its environment (see~\cite{iwls10:ch2,iwls10:ch3,iwls10:ch9}). 
An alternative formulation of FSM network synthesis was provided
by encoding the problem in the logic WS1S (Weak Second-Order Logic of 
1 Successor), which enables to characterize the set of permissible behaviors
at a node of a given network of FSMs by WS1S formulas~\cite{S1S-tcad2000},
corresponding to regular languages and so to effective operations on finite 
state automata.
\footnote{A detailed survey of previous work in this area can be found in~\cite{fsm1-book,villa12}.}

Another stream of contributions has been motivated by component-based design 
of parallel systems with an interleaving semantics (denoted in our exposition 
by the composition operator $\diamond$). The problem is stated 
by Merlin and Bochmann~\cite{boch83} as follows:
``Given a complete specification of a given module and 
the specifications of some submodules, the method described below provides 
the specification of an additional submodule that, together with the other 
submodules, will provide a system that satisfies the specification of the 
given module.''
The problem was reduced to solving equations or inequalities over process 
languages, which are usually prefix-closed regular languages represented 
by labeled transition systems.
A closed-form solution of the  inequality $A \diamond x \subseteq B$ over prefix-closed regular languages, written as $proj_x(A \diamond B) \setminus proj_x(A \diamond \overline{B})$
(where $proj_x$ is a projection over the alphabet of $x
$), was given in~\cite{boch83,haghverdi99}.\footnote{For a discussion about the maximality of this solution and for more 
references, we refer to~\cite{villa12}, Sec. 5.2.1.}
This approach to solve the equation 
$A \diamond x = B$ has been further extended to obtain restricted solutions 
that satisfy properties such as safety and liveness, or are restricted to 
be FSM languages, which need to be input-progressive and avoid divergence
(see~\cite{haghverdi99,bochmann-deds2013,villa12}).
The quotient problem has been investigated also for delay-insensitive processes
to model asynchronous sequential circuits, 
see~\cite{verhoeff-mpc1989,mallon99,negu-concur00}. 
Equations of the form $A \diamond x \leq B$ were defined, and their largest 
closed-form solutions were written as $x = \left(A \diamond B^{\sim}\right)^{\sim}$,  
where $(\cdot)^{\sim}$ is a suitable unary operation.

An important application from discrete control theory is the model matching problem:
design a controller whose composition with a plant matches a given specification 
(see~\cite{barrett-deds98,mm-tac01}).
Another significant application of the quotient computation
has been the protocol design problem, and in particular, the protocol 
conversion problem (see~\cite{lam-tse-1988,green-tcomm-1986,DBLP:conf/dac/PasseroneRS98,DBLP:conf/iccad/PasseroneAHS02,kumar-conv97,hallal-conv2000,fujita-aspdac2007,Castagnetti:SEFM-2015}).
Protocol converter synthesis has been studied also over a variant of 
Input/Output Automata (IOA,~\cite{lynch-cwi89}), called Interface Automata 
(IA,~\cite{deAlfaro2001,dealfaro-interface03}), yielding a similar quotient equation 
$A \diamond_{IA} x \subseteq B$ and closed-form solution 
$\left(A \diamond_{IA} B^{\bot}\right)^{\bot}$, where $\diamond_{IA}$
is an appropriate interleaving composition defined for interface automata, 
and $(\cdot)^{\bot}$ is again a unary operation~\cite{bhaduri08}.

Some research focused on modal specifications represented by automata whose 
transitions are typed with {\em may} and {\em must} modalities, as in~\cite{larsen-xinxin1990,raclet2011modal}, with a solution of the quotient
problem for nondeterministic automata provided in ~\cite{benes2013}.
It is outside the scope of this paper to address the quotient problem for real-time and hybrid systems (see~\cite{cassez2000,bouyer2011} for verification and control in such settings).


As seen above, the quotient problem was studied
by different research communities working on various application domains
and formalisms. Often similar formulations and solutions were reached
albeit obfuscated by the different notations and objectives of the synthesis 
process.
This motivated a concentrated effort to distill the core of the problem, 
modeling it as solving equations over languages of the form 
$A \parallel x \preceq B$, where $A$ and $B$ are known components and $x$ is unknown, 
$\parallel$ is a composition operator, and $\preceq$ is a conformance relation
(see~\cite{ucp-ieeeproc2015} and the monograph~\cite{villa12} for full 
accounts).  
The notion of language was chosen as the most basic formalism to specify 
the components of the equation, and language containment $\subseteq$ was 
selected as conformance relation.
Two basic composition operators were defined each encapsulating a family
of variants: synchronous composition ($\bullet$) modeling the classical 
step-lock coordination, and interleaving composition ($\diamond$)
modeling asynchrony by which components may progress at different rates
(there are subtle issues in comparing the two types, as mentioned 
in~\cite{kurshan-asynch99,fsmeq1-iwls2004}).
Therefore two language equations were defined:
$A \bullet x \subseteq B$ and $A \diamond x \subseteq B$, where
the details of the operations to convert alphabets according 
to the interconnection topologies are hidden in the formula.
It turned out that the largest solutions have the same structure, respectively,
$\overline{A \bullet \overline{B}}$ and $\overline{A \diamond \overline{B}}$. This led to investigate the algebraic properties required by the 
composition operators to deliver the previous largest closed-form solutions
to unify the two formulas~\cite{ucp-ieeeproc2015}. This effort assumed that the underlying objects were sets, and that their operations were given in terms of set operations. This work, thus, could not account for quotient computations in more complex theories, like interface automata.

As a parallel development, in recent years we have seen the growth of a
rigorous theory of system design based on the algebra of contracts
(see the monograph~\cite{BenvenisteContractBook}). In this theory, a strategic
role is played by assume-guarantee (AG) contracts, in which the \textit{missing
component problem} arises: when the given components are not capable of
discharging the obligations of the requirements, define a quotient operation
that computes the contract for a component, so that by its addition to the
original set the resulting system fulfills the requirements. The quotient of AG
contracts was completely characterized very recently by a closed-form solution
proved in~\cite{agquotient}. Once again, the syntax of the quotient has the form $\left(A \parallel B^{-1}\right)^{-1}$ for contracts $A$ and $B$ and standard contract operations.

In summary, even though
the concrete models of the components, composition operators, conformance
relations and inversion functions vary significantly across chosen models and
application domains, the quotient formulas have similar syntax across theories.

\subsection{Motivation and contributions}
The motivation of this paper is to \textit{propose the
underlying mathematical structure common to all these instances of quotient
computation to be able to derive directly the solution formula for any equation
satisfying the properties of this common structure.}

We show that we can compute the quotient by only 
assuming the axioms of a {\em preorder}, enriched with a binary operation of
{\em source multiplication} and a unary {\em involution} operation. 
In particular we introduce the new algebraic notion of {\em \structName{s}}
characterized by a condition, called {\em \propName}, which 
guarantees the existence of the solution and yields a closed form for it. 
Then we show that a number of theories in computer science meet this condition,
e.g., Boolean lattices, AG contracts,
and interface automata; so for all of them we are able to (re-)derive axiomatically the formulas that compute
their related quotients.
We also introduce the concept of {\em \compStructName{s}}
to deal with structures defined over multiple domains, 
and we show that the equations $A \bullet x \le B$ admit a solution also
over \compStructName{s}, generalizing the known solutions of equations 
on languages over multiple alphabets with respect to synchronous 
and interleaving composition, well studied in the literature.

\subsection{Organization}
The paper is structured as follows. Sec.~\ref{sec:preordHeaps} develops the basic
mathematical machinery of \structName{s}, whereas Sec.~\ref{sec:preordHeapInstances}
shows that various theories are \structName{s}. Sec.~\ref{sec:sievedHeaps}
introduces \compStructName{s}, whereas Sec.~\ref{sec:sieved-langeq} applies
them to equations over languages with multiple alphabets.
Sec.~\ref{sec:conclusions} concludes. Some proofs are omitted due to space constraints.

\section{\StructName{s}}
\label{sec:preordHeaps}

In this section we introduce an algebraic structure for which the existence of quotients is guaranteed. We show in Section~\ref{sec:preordHeapInstances} that many theories in computer science are instances of this concept. First we introduce the notation we will use:

    \begin{itemize}[leftmargin=*]
    \item Let $P$ be a set and let $\smult\colon P \times P \to P$ be a binary operation on $P$. For any element $a \in P$, we let $\smultL{a}\colon P \to P$ be the function $\smultL{a} = \smult \circ (a \times \id)$, where $\id$ is the identity operator and $(a \times \id)\colon P \to P^2$ is the unary function $(a \times \id)\colon b \mapsto (a, b)$. Similarly, we let $\smultR{a} = \smult \circ (\id \times a)$. If we call $\mu$ multiplication, $\mu_a$ is left multiplication by $a$, and $\mu^a$ is right multiplication by $a$.
    
    \item For any set $P$, we let the mapping $\flip\colon P \times P \to P \times P$ be $\flip(a,b) = (b , a)$ ($a, b \in P$).
    
    \item Consider a set $P$ and a binary relation  $\le$ on $P$. Then $ \le $ is a preorder
    if it is reflexive and transitive; i.e., for all $a, b$ and $c$ in $P$, we have 
    $a \le a~ \text{(reflexivity)}$ and if 
    $a \le b~ \text{and}~ b \le c ~\text{then} ~ a \le c$ (transitivity).
    If a preorder is antisymmetric, ($a \le b$ and $b \le a$ implies $a = b$), then it is a partial order.

    \item Let $(P, \le)$ be a preorder and let $a, b \in P$. If $a \le b$ and $b \le a$, we write $a \simeq b$.

    \item Let $F \colon P \to P$. We say that $F$ is monotonic or order-preserving if $a \le b \Rightarrow Fa \le Fb$ for all $a, b \in P$. Similarly, we say that $F$ is antitone or order-reversing if $a \le b \Rightarrow Fb \le Fa$ for all $a, b \in P$.
    
    \item Suppose that $L, R\colon P \to P$ are two monotonic maps on $P$. We say that $(L,R)$ form an adjoint pair, or that $L$ is the left adjoint of $R$ ($R$ is respectively the right adjoint of $L$), or that the pair $(L, R)$ forms a Galois connection when for all $b, c \in P$, we have $L b \le c$ if and only if $b \le R c$.
    
    \item Let $F, G\colon P \to P$ be functions on a preorder $P$. We say that $F \le G$ when $F a \le G a$ for all $a \in P$.
\end{itemize}

\subsection{The concept of \structName}

As we discussed in the introduction, many times in engineering and computer science one encounters expressions of the form $A \bullet x \le B$, and one wishes to solve for the largest $x$ that satisfies the expression. The symbols have different specific meanings in the various domains, yet in all applications we know, the syntax for computing the quotient always has the form $\overline{A \bullet \overline{B}}$, where $\overline{(\cdot)}$ is an involution (i.e., a unary operator which is its own inverse). To give meaning to the inequality, at a minimum we need a preorder and a binary operation; to give meaning to the quotient expression, we need to assume the existence of an involution. In all compositional theories, the refinement order has the connotation of specificity: if $a \le b$ then $a$ is a refinement of $b$. The binary operation is usually interpreted as composition. The product $a \bullet b$ is understood as the design obtained when operating both $a$ and $b$ in a topology given by the mathematical description of each component. The unary operation is sometimes understood as giving an external view on an object. If a component has mathematical description $a$, then $\overline{a}$ gives the view that the environment has of the design element. In Boolean algebras, this unary operation is negation. In interface theories, it's usually an operation which switches inputs and output behaviors.

We thus introduce an algebraic structure consisting of a preorder, a binary operation which is monotonic in both arguments, and an involution which is antitone. We have called the binary operation \emph{source multiplication} for reasons having to do with category theory: we will show that this operation serves as the left functor of an adjunction. Therefore, its application to an object of the preorder yields the \emph{source} of one of the two arrows in the adjunction. Why not simply call it multiplication? Because source multiplication together with the involution generate another binary operation. This second operation we call \emph{target multiplication} because its application to an object yields the \emph{target} of one of the arrows in the adjunction. The unary operation will simply be called \emph{involution}.

The algebraic structure will be called \emph{\structName}. The inspiration came from engineering design. In some design methodologies, design elements at the same level of abstraction are not comparable in the refinement order. Indeed, a refinement of a design element usually yields a design element in a more concrete layer. But we are placing all components under the same mathematical structure. This suggested the name \emph{heap}. We add the adjective \emph{preorder} simply to differentiate the concept from existing algebraic heaps. We are ready for the definition:

\begin{definition}\label{df:pgroup}
    A \emph{\structName} is a structure $(P, \le, \smult, \inv)$, where $(P, \le)$ is a
    preorder; $\smult\colon P \times P \to P$ is a binary operation on $P$, monotonic in both arguments,
    called \emph{source multiplication}; and $\inv\colon P \to P$ is an antitone operation on $P$ called \emph{involution}. These operations satisfy the following axioms:
    \begin{compactitem}
        \item A1: $\inv^2 = \id$.
        \item A2a (left \propName): $\smultL{a} \circ \inv \circ \smultR{a} \circ \inv \le \id \quad\quad(a \in P)$.
        \item A2b (right \propName): $\smultR{a} \circ \inv \circ \smultL{a} \circ \inv \le \id \quad\quad(a \in P)$.
    \end{compactitem} 
\end{definition}

\begin{note}\label{kjdhaoiuybdp9a}
    In Definition \ref{df:pgroup}, we did not assume commutativity in $\smult$. If $\smult$ is commutative, we have $\smult = \smult \circ \flip$, so $\smultL{a} = \smult \circ (a \times \id)= \smult \circ \flip \circ (a \times \id) = \smult \circ (\id \times a)  = \smultR{a}$. It follows that for a commutative \structName, axioms A2a and A2b become
    \begin{equation}\label{eq:commReg}(\smultL{a} \circ \inv)^2 \le \id.\end{equation} 
\end{note}

We have discussed all elements in the definition of a \structName, except for the \propName conditions. What are they? Consider left \propName: $\smultL{a} \circ \inv \circ \smultR{a} \circ \inv \le \id$. Let $b \in P$ and set $B = (\inv \circ \smultR{a} \circ \inv) (b)$. Left \propName means that $B$ satisfies the expression $\smult(a, x) \le b$. Similarly, set $C = (\inv \circ \smultL{a} \circ \inv) (b)$. Right \propName means that $C$ satisfies $\smult(x, a) \le b$. When $\smult$ is commutative, we of course have $B = C$.
We will soon show a surprising fact: the axioms of a \structName are sufficient to guarantee that $B$ and $C$ are in fact the largest solutions to both expressions, i.e., $B$ and $C$ are the quotients for left and right source multiplication, respectively.
We show this immediately after introducing an important binary operation called target multiplication, but first we consider an example.

\paragraph{Example.} Consider a Boolean lattice $B$. The lattice is clearly a preorder. Take the involution to be the negation operator. This is an antitone operator and satisfies A1: $\neg \neg b = b$ for all $b \in B$. Take source multiplication to be the meet of the lattice (i.e., logical AND). This operation is monotonic in the preorder. Since this source multiplication is commutative, the \propName conditions reduce to checking \eqref{eq:commReg}. For $a, b \in B$, we have
$(\mu_a \circ \gamma)^2 b = a \land \neg(a \land \neg b) = a \land (\neg a \lor b) = a \land b \le b$. Thus, the Boolean lattice satisfies the \propName conditions, making it a \structName.

\subsection{Target multiplication}\label{sc:targetMult}

For the rest of this section, let $(P, \le, \smult, \inv)$ be a \structName. We define the \emph{target multiplication} $\tmult\colon P \times P \to P$ as 
$
\tmult = \inv \circ \smult \circ (\inv \times \inv)
$. Since $\inv^2 = \id$ (axiom A1), we can also write $\smult = \inv \circ \tmult \circ (\inv \times \inv)$, i.e., the diagram
{\scriptsize
\begin{tikzcd}[sep = small]
    P \times P \arrow[d, "\inv \times \inv\,"', leftrightarrow] 
    \arrow[r, "\smult"]
    & P \\
    P\times P \arrow[r, "\tmult"] & P \arrow[u, "\inv"', leftrightarrow]
\end{tikzcd}}
commutes.

We could have defined a \structName in terms of target multiplication instead of source multiplication. The two operations are closely linked. In fact, we will see in the next section that these operations form an adjoint pair.

\paragraph{Example.} We showed that Boolean lattices are \structName{s}. For $B$ a Boolean lattice and $a, b \in B$, we have $\tmult(a, b) = \gamma \circ \smult (\gamma a, \gamma b) = \neg (\neg a \land \neg b) = a \lor b$. This suggests that the relation between source and target multiplications is a generalization of De Morgan's identities for Boolean algebras.

We will use the following identities: for $a \in P$,
\begin{align}\label{xhjdw98dsqn}
    \begin{aligned}
        \smultL{a} &= \inv \circ \tmult \circ (\inv \times \inv) \circ (a \times \id) =
        \inv \circ \tmult \circ (\inv a \times \id) \circ \inv = 
        \inv \circ \tmultL{\inv a} \circ \inv \quad\quad \text{and}
        \\
        \smultR{a} &= \inv \circ \tmult \circ (\inv \times \inv) \circ (\id \times a) =
        \inv \circ \tmult \circ (\id \times \inv a) \circ \inv = 
        \inv \circ \tmultR{\inv a} \circ \inv.
    \end{aligned}
\end{align}

\subsection{Solving inequalities in \structName{s}}

For $a, b \in P$, we are interested in the conditions under which we can find
the largest $x \in P$ such that $\smult(a,x) \le b$. The following theorem says that source multiplication in a \structName is ``invertible.''

\begin{theorem}\label{9721w91hja}
    Let $(P, \le, \smult, \inv)$ be a \structName and let $\tmult$ be its target multiplication. Then 
    for $a \in P$, $(\smultL{a}, \tmultR{\inv a})$ and $(\smultR{a}, \tmultL{\inv a})$
    are adjoint pairs.
\end{theorem}

\begin{proof}
    Let $b, c \in P$ with $b \le \tmultR{\inv a}(c)$. We have $\smultL{a}(b) \le (\smultL{a} \circ \tmultR{\inv a}) (c) = (\smultL{a} \circ \inv \circ \smultR{a} \circ \inv)(c) \le c$, by left \propName (by A2a).
    
    Conversely, assume that $\smultL{a}(b) \le c$. Then
    \begin{align*}
        \smultL{a} \circ \inv^2 (b) &\le c & \text{(by A1)}\\
        \inv \circ (\smultL{a} \circ \inv) (\inv b) &\ge \inv (c) \\
        (\smultR{a} \circ \inv) \circ (\smultL{a} \circ \inv) (\inv b) &\ge (\smultR{a} \circ \inv) (c) \\
        (\inv b) &\ge (\smultR{a} \circ \inv) (c) & \text{(by A2b)} \\
        b &\le  (\inv \circ \smultR{a} \circ \inv) (c) = \tmultR{\inv a} (c). & \text{(by A1)}
    \end{align*}
    The adjointness of $(\smultR{a}, \tmultL{\inv a})$ follows from a similar reasoning.
\end{proof}

The fact that $(\smultL{a}, \tmultR{\inv a})$ is an adjoint pair means that left source multiplication by $a$ is ``inverted'' by right target multiplication by $\inv a$, i.e.,
\[
\mu(a, x) \le b \quad\text{if and only if}\quad x \le \tau(b, \inv a).
\]
In other words, the largest solution of $\mu(a, x) \le b$ is $x = \tau(b, \inv a)$. Using the familiar multiplicative notation for source multiplication, and $(\cdot)/a = \tmultR{\inv a}$ for ``right division by $a$,'' we have shown that the largest solution of $a x \le b$ is $x = b / a$. Calling $a \backslash (\cdot) = \tmultL{\inv a}$ ``left division by $a$,'' we have shown that the largest solution of $x a \le b$ is $x = a \backslash b$. These two divisions are related as follows:

\begin{corollary}[Isolating the unknown] 
    Let $P$ be a \structName and $a, x, y \in P$. Then $y \le a / x$ if and only if $x \le y \backslash a$.    
\end{corollary}
\begin{proof}
    By two applications of Theorem \ref{9721w91hja}, we obtain
    $y \le a / x = \tmultR{\inv x}(a) \Leftrightarrow
    \smult(x,y) \le a \Leftrightarrow
    x \le \tmultL{\inv y}(a) = y \backslash a$.
\end{proof}

Theorem \ref{9721w91hja} is our main result. It shows that \structName{s} have sufficient structure for the computation of quotients. When we prove that a structure is a \structName, this theorem immediately yields the existence of an adjoint for multiplication, and its closed form.

In general, to show that a theory is a \structName, we must identify its involution and source multiplication. Then we have to verify the \propName conditions. How difficult is that? Our original problem was identifying the largest $x$ satisfying $\mu(a,x) \le b$ for some notion of multiplication $\mu$, involution $\gamma$, and preorder $\le$. As we discussed, left \propName requires that $\tmult^{\inv a} b$ satisfies the inequality $\smult(a,x) \le b$, and right \propName requires that
$\tmult_{\inv a} b$ satisfies $\smult(x, a) \le b$. What the theorem tells us is that they are \emph{the largest solutions} to $\smult(a,x) \le b$ and $\smult(x, a) \le b$, respectively. In other words, the theorem saves us the effort of making an argument for the optimality of the solutions.

Theorem \ref{9721w91hja} also suggests the following observation. For a given $a \in P$, we have adjoint pairs $(\mu_a, \tau^{\gamma a})$ and $(\mu^a, \tau_{\gamma a})$. As we noticed, this means we can find the largest $x$ such that $\mu(a, x) \le b$ or $\mu(x,a) \le b$. But it also means that we can find \emph{the smallest} $x$ such that $b \le \tau(a, x)$ or $b \le \tau(x, a)$. This is because, $\mu_{\gamma a}$ is the left adjoint of $\tau^{a}$, and $\mu^{\gamma a}$ is the left adjoint of $\tau_{a}$. 
For all examples we will discuss, source multiplication plays the role of the usual composition operation of the theory. But \structName{s} make it clear that $\mu$ and $\tau$ are closely related operations. In fact, \structName{s} generalize De Morgan's identities (see section \ref{sc:targetMult}). Thus, while inequalities of the form $\mu(a, x) \le b$ are more common in the literature, \structName{s} indicate that we can also solve inequalities of the form $b \le \tau(a, x)$. As we will see, for some theories there is clear understanding of how target multiplication can be used, but for others its use is unknown.

\paragraph{Example.} In the case of a Boolean lattice $B$, what is the quotient? We showed in previous examples that $B$ is a \structName, and we identified its target multiplication. For $a, b \in B$, we can write an expression of the form $\mu(a, x) \le b$. By Theorem~\ref{9721w91hja}, we know the largest $x$ that satisfies this expression is $\tau^{\gamma a} b = \tau(b, \neg a) = b \lor \neg a$, i.e., the quotient is the implication $a \to b$.

\subsection{\StructName{s} with identity}

In the definition of a \structName, we did not assume that source multiplication has an identity. Here we consider briefly what happens when it does. Multiplicative identities are common, and in fact, there exists a multiplicative identity in all compositional theories we know.

Suppose $P$ is a \structName and $e \in P$ is a left identity for source
multiplication, i.e., $\smultL{e} \simeq \id$. By Theorem \ref{9721w91hja},
$(\id, \tmultR{\inv e})$ is an adjoint pair. The right adjoint of $\id$ is
$\id$. Since adjoints are unique up to isomorphism, $\tmultR{\inv e} \simeq
\id$. This means that $\inv e$ is a right identity element for $\tmult$.
Moreover, in view of \eqref{xhjdw98dsqn}, $\tmultL{\inv e} \simeq \id$. By
Theorem \ref{9721w91hja}, $(\smultR{e}, \id)$ is an adjoint pair. By the same
reasoning just followed, we must have $\smultR{e} \simeq \id$. We record this
result:

\begin{corollary}\label{kjhdpoqlkjndo}
    Let $(P, \le, \smult, \inv)$ be a \structName. If $e \in P$ is a left (or right) identity for source multiplication, it is a double-sided identity for source multiplication, and $\inv e$ is a double-sided identity for target multiplication. 
    Analogously, if $e \in P$ is a left (or right) identity for target multiplication, it is a double-sided identity for target multiplication, and $\inv e$ is a double-sided identity for source multiplication. 
\end{corollary}
    
\paragraph{Example.} Let $B$ be a Boolean lattice. The top element of the
lattice, usually denoted 1, is an identity for source multiplication: $1 \land a
= a$ for all $a \in B$. The previous corollary tells us that $\neg 1 = 0$ is a
double sided identity for target multiplication, which we identified to be
disjunction.

\section{Additional instances of \structName{s}}
\label{sec:preordHeapInstances}

As described in Section \ref{sec:preordHeaps}, as soon as we verify that a theory is a \structName, we know how to compute quotients for that theory. Here we show that assume-guarantee (AG) contracts and interface automata are \structName{s}. In both cases, we first define the algebraic aspects of the theory, and then we proceed to show that it is a \structName, which involves verifying the axioms of Definition \ref{df:pgroup}. After we do this, we invoke Theorem \ref{9721w91hja} to express its quotient in closed-form. The literature for both theories is large, and we only discuss them algebraically. To learn about their uses and the design methodologies based on them, we suggest \cite{BenvenisteContractBook} and \cite{deAlfaro2001}.

\subsection{AG contracts}

Assume-guarantee contracts are an algebra and a methodology to support compositional system design and analysis. Fix once and for all a set $B$ whose elements we call behaviors. Subsets of $B$ are referred to as behavioral properties or trace properties. An AG contract is a pair of properties $C = (A, G)$ satisfying $A \setunion G = B$. Contracts are used as specifications: a component adheres to contract $C$ if it meets the guarantees $G$ when instantiated in an environment that satisfies the assumptions $A$. The specific form of these properties is not our concern now; we are only interested in the algebraic definitions. The algebra of assume-guarantee contracts was introduced by R. Negulescu \cite{negu-concur00} (there called \emph{process spaces}) to deal with assume-guarantee reasoning for concurrent programs. The algebra was reintroduced, together with a methodology for system design, by Benveniste et al.~\cite{multViewpoint} to apply assume-guarantee reasoning to the design and analysis of any engineered system. Now we describe the operations of this algebra.

For $C' = (A', G')$ another contract, the partial order of AG contracts, called \emph{refinement}, is given by $C \le C'$ when $G \subseteq G'$ and $A \supseteq A'$. The involution of AG contracts, called reciprocal, is given by $\inv C = (G, A)$. This operation is clearly antitone and meets axiom A1. Source multiplication is contract composition: $\smult(C, C') = \left(A \setint A' \setunion \neg (G \setint G'), G \setint G'\right)$. This operation yields the tightest contract obeyed by the composition of two design elements, each obeying contracts $C$ and $C'$, respectively. Composition is monotonic in the refinement order of AG contracts. We need to verify the \propName conditions. Since source multiplication for AG contracts is commutative, we verify \eqref{eq:commReg}:
\begin{align*}
    (\smultL{C} &\circ \inv)^2 C' =
    (\smultL{C} \circ \inv) \circ (\smultL{C}) (G', A') =
    \smultL{C} (G \setint A', A \setint G' \setunion \neg (G \setint A')) \\
    & = (A \setint G \setint A' \setunion \neg G \setunion \neg (A \setint G' \setunion \neg A'),
    G \setint (A \setint G' \setunion \neg A'))
    \\ 
    & = (A \setint A' \setunion \neg G \setunion \neg A \setint A' \setunion \neg G' \setint A',
    G \setint (A \setint G' \setunion \neg A'))
    \\ &=
    (A' \setunion \neg G, G \setint (A \setint G' \setunion \neg A')) \le (A', G') = C',
\end{align*}
where in the last step we used the fact that $\neg A' \subseteq G'$, which follows from $A' \setunion G' = B$. We conclude that AG contracts satisfy the \propName conditions, and thus have \structName structure.

What is target multiplication for AG contracts? From its definition, we have
$\tau(C, C') = \gamma \circ \mu \circ (\gamma C, \gamma C') = \gamma \circ \mu \left( (G, A), (G', A') \right) =
\left( A \setint A', G\setint G' \setunion \neg (A\setint A') \right)$. This is an operation on contracts called \emph{merging}. One of the main objectives of the theory of assume-guarantee contracts is to deal with \emph{multiple viewpoints}, i.e., a multiplicity of design concerns, each having a contract representing the specification for that concern (e.g., functionality, timing, etc.). In \cite{contractMerging}, it is argued that the operation of merging is used to bring multiple viewpoint specifications into a single contract object. 

Since AG contracts are \structName{s}, we get their quotient formulas from Theorem \ref{9721w91hja}. The adjoint of $\smultL{C'}$ is $\tau^{\gamma C'} = \inv \circ \smultR{C'} \circ \inv$. Applying this to $C$ yields $\tau^{\gamma C'} (C) = \inv \circ \smultR{C'} (G, A) =
(A \setint G', G \setint A' \setunion \neg (A \setint G'))$. This closed-form expression for the quotient of AG contracts was first reported in \cite{agquotient}. Also by Theorem \ref{9721w91hja}, the left adjoint of merging by a fixed contract $C'$ is the operation $\mu(C, \gamma C ') = \mu\left( (A, G), (G', A') \right) = \left(A \setint G' \setunion \neg(G \setint A'), G \setint A'\right)$. This operation was recently introduced under the name of \emph{separation} in \cite{contractMerging}.

\subsection{Interface automata}
\label{sec:interface-automata}

We show that Interface Automata as introduced in~\cite{deAlfaro2001} have \structName structure. To achieve this result, we first provide the relevant definitions for interface automata. All definitions match those of ~\cite{deAlfaro2001}, except for our definition of alternating simulation for interface automata.

\newcommand{\actset}[0]{\mathcal{A}}
\newcommand{\stepset}[0]{\mathcal{T}}

An interface automaton $P = \langle V_P, V_P^{\text{init}}, \actset_P^I, \actset_P^O, \actset_P^H, \stepset_P \rangle$ consists of the following elements:
\begin{compactitem}
\item $V_P$ is a set of states.
\item $V_P^{\text{init}} \subseteq V_P$ is a set of initial states. Following~\cite{deAlfaro2001}, we require that $V_P^{\text{init}}$ contains at most one state. 
\item $\actset_P^I, \actset_P^O$, and $\actset_P^H$ are mutually disjoint sets of input, output, and internal actions. We denote by $\actset_P = \actset_P^I \setunion \actset_P^O \setunion \actset_P^H$ the set of all actions.
\item $\stepset_P \subseteq V_P \times \actset_P \times V_P$ is a set of steps.
\end{compactitem}

Following~\cite{deAlfaro2001}, if $a \in \actset_P^I$ (resp. $a \in \actset_P^O$, $a \in \actset_P^H$), then $(v,a,v')$ is called an input (resp. output, internal) step.  We denote by $\stepset_P^I$ (resp. $\stepset_P^O$, $\stepset_P^H$) the set of input (resp. output, internal) steps.
An action $a \in \actset_P$ is enabled at a state $v \in V_P$ if there is a step $(v,a,v') \in \stepset_P$ for some $v' \in V_P$. We indicate by $\actset_P^I(v), \actset_P^O(v), \actset_P^H(v)$ the subsets of input, output, and internal actions that are enabled at the state $v$, and we let $\actset_P(v) = \actset_P^I(v) \setunion \actset_P^O(v) \setunion \actset_P^H(v)$.


\begin{definition}
\label{definition:ia:product}
If $P$ and $Q$ are interface automata, let $\shared(P,Q) = (\actset_P^I \setint \actset_Q^O) \setunion (\actset_P^O \setint \actset_Q^I)$. The product $P \otimes Q$ is the interface automaton with the following constituents:
$V_{P \otimes Q} = V_P \times V_Q$, 
$V_{P \otimes Q}^{\text{init}} = V_P^{\text{init}} \times V_Q^{\text{init}}$, 
$\actset_{P \otimes Q}^I = (\actset_P^I \setunion \actset_Q^I) \setminus \shared(P,Q)$, 
$\actset_{P \otimes Q}^O = (\actset_P^O \setunion \actset_Q^O) \setminus \shared(P,Q)$, 
$\actset_{P \otimes Q}^H = \actset_P^H \setunion \actset_Q^H \setunion \shared(P,Q) \setminus (\actset_{P \otimes Q}^I \setunion \actset_{P \otimes Q}^O)$, and
\begin{align*}
\stepset_{P \otimes Q} & = \setArg{ ((v,u),a,(v',u))}{(v,a,v') \in \stepset_P \wedge a \in \actset_P \setminus \actset_Q \wedge u \in V_Q } \\
& \setunion \setArg{ ((v,u),a,(v,u'))}{ (u,a,u') \in \stepset_Q \wedge a \in \actset_Q \setminus \actset_P  \wedge v \in V_P } \\
& \setunion \setArg{ ((v,u),a,(v',u')) }{ (v,a,v') \in \stepset_P \wedge (u,a,u') \in \stepset_Q \wedge a \in \actset_P \setint \actset_Q  }.
\end{align*}

\end{definition}

We call illegal those states of the product in which one of the interface automata can take a step through a shared action, but the other can't. These states are removed from the product in the definition of composition of interface automata. Given two composable interface automata $P$ and $Q$, the set $\illegal(P,Q)$ $\subseteq$ $V_P \times V_Q$ of illegal states of $P \otimes Q$ is given by
\[
\illegal(P,Q) = \setArg{(v,u) \in V_P \times V_Q }{ \exists a \in \shared(P,Q). 
    \left( 
        \begin{array}{c}
            a \in \actset_P^O(v) \wedge a \notin \actset_Q^I(u) \\
            \vee\\
            a \in \actset_Q^O(u) \wedge a \notin \actset_P^I(v) \\
        \end{array}
    \right)
}.
\]

An environment for an interface automaton $R$ is an interface automaton $E$ such that $E$ is composable with $R$, $E$ is nonempty, $\actset_E^I = \actset_R^O$, and $\illegal(R,E) = \emptyset$. A legal environment for the pair $(P,Q)$ is an environment for $P \otimes Q$ such that no state in $\illegal(P,Q) \times V_E$ is reachable in $(P \otimes Q) \otimes E$. We say that a pair $(v,u) \in V_P \times V_Q$ of states is compatible if there is an environment $E$ for $P \otimes Q$ such that no state in $\illegal(P,Q) \times V_E$ is reachable in $(P \otimes Q) \otimes E$ from the state $\{(v,u)\} \times V_E^{\text{init}}$. Two interface automata $P$ and $Q$ are compatible if the initial state $(v,u) \in V_P^{init} \times V_Q^{init}$ is compatible. We write $\comp(P,Q)$ for the set of compatible states of $P \otimes Q$. With these notions, we can define parallel composition for interface automata.

Given two compatible interface automata $P$ and $Q$, the composition $P \parallel Q$ is an interface automaton with the same action sets as $P \otimes Q$. The states are $V_{P \parallel Q} = \comp(P,Q)$; the initial states are $V_{P \parallel Q}^{\text{init}}$ $=$ $V_{P \otimes Q}^{\text{init}} \setint \comp(P,Q)$; and the steps are $\stepset_{P \parallel Q}$ $=$ $\stepset_{P \otimes Q}$ $\setint$ $(\comp(P,Q) \times \actset_{P \parallel Q} \times \comp(P,Q))$.

Let $v \in V_P$, the set $\epsclosure_P(v)$ is the smallest set $U \subseteq V_P$ such that $v \in U$ and if $u \in U$ and $(u,a,u') \in \stepset_P^H$, then $u' \in U$. Moreover, we let
\[
\exten_P^O(v) = \bigcup_{u \in \epsclosure(v)} \actset_P^O(u)
\quad\text{and}\quad
\exten_P^I(v) = \bigcup_{u \in \epsclosure(v)} \actset_P^I(u)
\]
be the sets of externally enabled output and input actions, respectively, at $v$. And for all externally enabled input and output actions $a \in \exten_P^I(v) \setunion \exten_P^O(v)$, we let
\[
\extdest_P(v,a) = \setArg{ u' }{\exists (u,a,u') \in \stepset_P.\,\, u \in \epsclosure_P(v) }.
\]
With these notions, we can define an alternating simulation between interface automata.

\begin{definition}
\label{definition:ia:alternating-simulation}
Consider two interface automata $P$ and $Q$. A binary relation $\preceq \, \subseteq V_Q \times V_P$ is an alternating simulation from $Q$ to $P$ if for all states $u \in V_Q$ and $v \in V_P$ such that $u \preceq v$, the following conditions hold:
\\(a) $\exten_P^I(v) \subseteq \exten_Q^I(u), \quad\quad \exten_Q^O(u) \subseteq \exten_P^O(v)$.
\\(b) For all actions $a \in \exten_Q^O(u)$ and all states $u' \in \extdest_Q(u,a)$, there is a state $v' \in \extdest_P(v,a)$ such that $u' \preceq v'$ and for all actions $a \in \exten_P^I(v)$ and all states $v' \in \extdest_P(v,a)$, there is a state $u' \in \extdest_Q(u,a)$ such that $u' \preceq v'$.
\end{definition}

Now we use the notion of alternating simulation to establish a preorder for interface automata:
the interface automaton $Q$ refines the interface automaton $P$, written $Q \preceq P$, if $\actset_P^I \subseteq \actset_Q^I$, $\actset_P^O \supseteq \actset_Q^O$, and there is an alternating simulation $\preceq$ from $Q$ to $P$, a state $v \in V_P^{\text{init}}$, and a state $u \in V_Q^{\text{init}}$ such that $u \preceq v$.

Let $P = \langle V_P, V_P^{\text{init}},A_P^I,A_P^O,A_P^H,T_P \rangle$ be an interface automaton. The \mirror of $P$, denoted $P^{\transp}$, is given by $P^{\transp} = \langle V_P, V_P^{\text{init}},A_P^O,A_P^I,A_P^H,T_P \rangle$.
The \mirror operation is clearly an involution, i.e., $\left(P^{\transp}\right)^{\transp} = P$. Let the source multiplication $\mu$ be the parallel composition of interface automata, $\inv$ be the \mirror operation, and let the preorder be refinement. We state the main claim of this section:

\begin{proposition}
\label{proposition:ia:preordHeap}
A theory of interface automata is a \structName.
\end{proposition}

Since interface automata have \structName structure, for given interface automata $P$ and $Q$, Theorem \ref{9721w91hja} enables us to find largest solutions $R$ for equations of the form $\mu(Q, R) \le P$. The quotient for interface automata was first reported in \cite{bhaduri08}.
Now that we know interface automata have \structName structure, we can ask: what is target multiplication for interface automata? The operation is given by $\tau(P, Q) = \left(P^\transp \parallel Q^\transp\right)^\transp$. We propose to call this operation \emph{merging} in analogy to the case of AG contracts. Similarly, by Theorem \ref{9721w91hja}, merging by fixed $Q$, $\tau_{Q}$, has a left adjoint given by
$\mu^{\gamma Q} (P) = P \parallel Q^\transp$. For the same reason, we propose to call this binary operation \emph{separation}. 
In AG contracts, merging and separation are used to handle multiple viewpoints in a design. To the best of our understanding, the notion of handling multiple design viewpoints has not been discussed for interface automata. Maintaining the analogy to AG contracts, we suspect that merging and separation here defined provide interface automata the ability to handle these multiple viewpoints. Exploring this idea is material for future work.

\section{\CompStructName{s}}
\label{sec:sievedHeaps}

Some theories in computer science require manipulating objects which are not defined over the same domain. For example, consider a language $L_1$ defined over an alphabet $\Sigma_1$. Let $\Sigma_2$ be another alphabet for which $L_2$ is a language. The powerset of a set is a Boolean lattice, so we have two \structName{s} $P_{\Sigma_1} = 2^{\Sigma_1^*}$ and $P_{\Sigma_2} = 2^{\Sigma_2^*}$ whose source multiplications and involutions are intersection and negation ($*$ is the Kleene star---we will define operations carefully in the section on languages). With the theory of \structName{s}, we know how to solve inequalities for $P_{\Sigma_1}$ and for $P_{\Sigma_2}$. Suppose we define an operation that allows us to compose $L_1 \in P_{\Sigma_1}$ with $L_2 \in P_{\Sigma_2}$. How do we solve inequalities involving $L_1$ and $L_2$ then? These languages belong to different \structName{s}. It is natural to define such an operation by mapping $L_1$ and $L_2$ to a common \structName, which by definition, has its own notion of source multiplication. We need a notion of mapping between \structName{s}:

\begin{definition}
    Let $(P, \le, \smult, \inv)$ and $(P', \le', \smult', \inv')$ be two \structName{s}. A \structName homomorphism $f \colon  P \to P'$ is an order-preserving map which commutes with the source multiplications and involutions, i.e.,
    {\scriptsize
    $
        \begin{tikzcd}[sep = small]
            P \times P \arrow[r, "f \times f"] \arrow[d, "\mu"'] & P' \times P' \arrow[d, "\mu'"] \\
            P  \arrow[r, "f "] & P'
        \end{tikzcd} 
    $} and 
    {\scriptsize
    $
        \begin{tikzcd}[sep = small]
            P  \arrow[r, "f"] \arrow[d, "\inv"'] & P' \arrow[d, "\inv'"] \\
            P  \arrow[r, "f "] & P'
        \end{tikzcd}
    $} commute.
\end{definition}

\StructName{s} $P_{\Sigma_1}$ and $P_{\Sigma_2}$ are indexed by alphabets. The common \structName where $L_1$ and $L_2$ can be mapped is determined by $\Sigma_1$ and $\Sigma_2$. As we will see in the next section, one option is to say that they generate the alphabet $\Sigma_c = \Sigma_1 \setunion \Sigma_2$, and we can define maps $\iota_1 \colon P_{\Sigma_1} \to P_{\Sigma_c}$ and $\iota_2 \colon P_{\Sigma_2} \to P_{\Sigma_c}$ that embed languages over $\Sigma_1$ and $\Sigma_2$ to those defined under $\Sigma_c$. This observation tells us that we can use a structure $S$ in order to index \structName{s}; this structure must have a binary operation defined in it. This operation will fulfill the role of identifying the alphabets where two languages can meet. Call this structure $S$, and let $\cdot$ be its binary operation. If we have two languages defined over the same alphabet, we should not need to move to another alphabet to compute the source multiplication of the two languages; thus, the binary operation of $S$ should be idempotent. We will also require the operation to be commutative since it makes no difference whether we go to the language generated by $\Sigma_1$ and $\Sigma_2$ or to that generated by $\Sigma_2$ and $\Sigma_1$. A similar reasoning leads us to require associativity. Thus, $S$ is endowed with an associative, commutative, idempotent binary operation, which means it is a semilattice. We make the choice to interpret it as an upper semi-lattice because we have the intuition that the languages generated by two smaller languages should be larger than any of the two, but this interpretation does not impose any algebraic limitations: an upper semilattice can be turned into a lower semilattice simply by flipping it upside-down.

We introduce the notion of a sieved, \structName (\compStructName, for short) that allows us to move objects between different domains of definition or different levels of abstraction. A \compStructName is a collection of \structName{s} indexed by an upper semilattice $S$ together with mappings between the \structName{s}. We call these mappings concretizations. An upper semilattice can be interpreted as a partial order: for $a, b \in S$, we say that $a \le a b$. Thus, the shortest definition for a \compStructName is that it is a functor from the preorder category $S$ to {\bf{PreHeap}}, the \structName category, whose objects are \structName{s} and whose arrows are \structName homomorphisms. We will give a longer definition. But first, why the adjective sieved? A \compStructName consists of a collection of \structName{s} and maps between them. We interpret these \structName{s} as structures containing varying amounts of detail about an object. This varying granularity motivated the name. This is the definition of this composite structure:


%
%
%

\begin{definition}\label{df:sievedHeaps}
Let $\sgroup$ be a semilattice. Let $\{(P_x, \le_x, \mu_x, \inv_x)\}_{x \in \sgroup}$ be a collection of \structName{s} such that for every $x, y, z \in \sgroup$ we have a unique \structName homomorphism $\iota\colon  P_x \to P_{xy}$ referred to as a concretization and making 
{\scriptsize $
    \begin{tikzcd}[sep = small]
        P_{xy} \arrow[dr, "\iota'"] & \\ 
        P_x \arrow[u, "\iota"]
        \arrow[r, "\iota''"]& 
        P_{xyz}
    \end{tikzcd}
$}
commute. 
We require the concretization $\iota \colon  P_x \to P_x$ to be the identity. Let $P = \oplus_{x \in \sgroup} P_x$, where $\oplus$ stands for disjoint union. We call $(P, \le, \mu, \inv)$ an $\sgroup$-\compStructName, where $\smult\colon  P\times P \to P$ is an operation called source multiplication, and $\inv \colon  P \to P$ is called involution. Let $a \in P_x$ and $b \in P_y$, and let $\iota_x\colon  P_x \to P_{xy}$ and $\iota_y\colon  P_y \to P_{xy}$ be concretizations.
These operations are given by
\begin{align*}
    \mu(a, b) = \mu_{xy} \left(\iota_{x}(a), \iota_{y}(b)\right) \quad\text{and} \quad 
    \inv(a) = \inv_{x}(a).
\end{align*}
Moreover, we say that $a \le b$ if and only if there exists $z \in S$ and concretizations $\iota\colon  P_x \to P_z$ and $\iota'\colon  P_y \to P_z$ such that 
$\iota (a) \le_{z} \iota' (b)$, where $\le_{z}$ is the preorder of $P_{z}$.
\end{definition}

Target multiplication $\tau$ for $P$ is defined in a similar way: $\tau(a, b) = \tau_{xy} \left(\iota_{x}(a), \iota_{y}(b)\right)$, where $\tau_{xy}$ is the target multiplication of the \structName $P_{xy}$.

\subsection{\CompStructName{s} are \structName{s}}
\label{sec:sievedHeapsArePreordered}

Now we show that a \compStructName is itself a \structName. To do this, we must show that the relation $\le$ over \compStructName{s} is a preorder, that source multiplication defined for a \compStructName is monotonic, that its involution is antitone, and that it meets the \propName conditions. The following statements show that \compStructName{s} have these properties.

\begin{lemma}\label{lm:kqoasxnh}
    The relation $\le$ on an $S$-\compStructName $P$ is a preorder.
\end{lemma}

\begin{proof}
    Reflexivity. Let $a \in P_x$. Let $\iota$ be the concretization $\iota\colon P_x \to P_x$. Then $\iota a \le_x \iota a$ because $\le_x$ is a preorder in $P_x$; this means that $a \le a$ in $P$.
    
    Transitivity. Let $b \in P_y$ and $c \in P_z$ and suppose that $a \le b$ and $b \le c$. Then there exist $v, w \in S$ such that $\iota_x a \le_v \iota_y b$ and $\iota_y' b \le_w \iota_z c$, where
    the diagram
    {\scriptsize $
    \begin{tikzcd}[sep = small]
            &P_{v} \arrow[r, "\iota_{v}"] & P_{vw}
            &P_{w} \arrow[l, "\iota_{w}"']&
            \\
            &
            P_x \arrow[u, "\iota_x"]
            & P_{y} \arrow[ul, "\iota_y"]
            \arrow[ur, "\iota_y'"']
            &
            P_z \arrow[u, "\iota_{z}"']
            & 
    \end{tikzcd}
    $}
    shows the relevant concretization maps (these diagrams commute per Definition \ref{df:sievedHeaps}). We obtain immediately 
    $\iota_v \circ \iota_x a \le_{vw} \iota_v \circ \iota_y b$ and $\iota_w \circ \iota_y' b \le_{vw} \iota_w \circ \iota_z c$. From the diagram, $\iota_v \circ \iota_y = \iota_w \circ \iota_y'$, which means that
    $\iota_v \circ \iota_x a \le_{vw} \iota_w \circ \iota_z c$, which means that $a \le c$.
    \end{proof}

\begin{lemma}\label{lm:kdybqoduny}
    Source multiplication on $P$ is monotonic in both arguments.     
\end{lemma}

\begin{proof}
    Let $a, b, c \in P$ with $a \le c$. Suppose that $a \in P_x$, $b \in P_y$, and $c \in P_z$. Since $a \le c$, there exist $u \in S$ such that $\iota_x a \le_u \iota_z c$ for concretizations $\iota_x\colon P_x \to P_u$ and $\iota_z\colon P_z \to P_u$. Note that this means there exist $u', u'' \in S$ such that $u = x u'$ and $u = z u''$. But this implies that $u y = x y u'$ and $u y = y z u''$. Thus, there exist concretizations $\iota_{xy}\colon P_{xy} \to P_{uy}$ and $\iota_{yz}\colon P_{yz} \to P_{u y}$, and
    \begin{equation}\label{eq:skhgsiu} {\scriptsize 
\begin{tikzcd}[sep = small]
        & & P_{y} \arrow[dl, "\iota_y'"'] \arrow[d, "\iota_y"] \arrow[dr, "\iota_y''"]&& \\
        &P_{xy} \arrow[r, "\iota_{xy}"'] & P_{uy}
        &P_{yz} \arrow[l, "\iota_{yz}"]&
        \\
        &
        P_x \arrow[u, "\iota_x'"] \arrow[r, "\iota_x"]
        & P_{u} \arrow[u, "\iota_u"']
        &
        P_z \arrow[u, "\iota_{z}'"'] \arrow[l, "\iota_z"']
        &
\end{tikzcd}}
\end{equation}
commutes. Since $a \le c$, we have
    \begin{align}\label{eq:jdiuq}
    \mu_{uy}\left(\iota_u \circ \iota_x a, \iota_y b\right) \le_{uy}
    \mu_{uy}\left(\iota_u \circ \iota_z c, \iota_y b \right).
    \end{align}

By the commutativity of the diagram, $\iota_y = \iota_{xy}\circ \iota_y' = \iota_{yz}\circ \iota_y''$ and $\iota_u \circ \iota_x = \iota_{xy} \circ \iota_x'$ and $\iota_u \circ \iota_z = \iota_{yz} \circ \iota_z'$. Using these identities, we can rewrite \eqref{eq:jdiuq} as
\begin{align*}
    \mu_{uy}\left(\iota_{xy} \circ \iota_x' a, \iota_{xy}\circ \iota_y' b\right) &\le_{uy}
    \mu_{uy}\left(\iota_{yz} \circ \iota_z' c, \iota_{yz}\circ \iota_y'' b \right), \quad \text{which implies that} \\
    \iota_{xy} \circ \mu_{xy}\left( \iota_x' a, \iota_y' b\right) &\le_{uy}
    \iota_{yz} \circ \mu_{yz}\left(\iota_z' c, \iota_y'' b \right) \text{ and thus }
    \iota_{xy} \circ \mu\left( a, b\right) \le_{uy}
    \iota_{yz} \circ \mu\left(c, b \right)
    .
\end{align*}
This shows that $\mu\left( a, b\right) \le \mu\left(c, b \right)$. Monotonicity in the second argument is proved in the same way.
\end{proof}

\begin{theorem}\label{kiybd9qybo}
    An $S$-\compStructName $P$ is a \structName.
\end{theorem}

\begin{proof}
    By lemma \ref{lm:kqoasxnh}, we know that $(P, \le)$ is a preorder. By lemma \ref{lm:kdybqoduny}, we know that source multiplication for $P$ is monotonic. From the definition of involution $\inv$ for $P$, it is immediate that this operation is antitone and that $\inv^2 = \id$. We must show the \propName conditions. Let $a \in P_x$ and $b \in P_y$. Using the notation of \eqref{eq:skhgsiu}, we have
    $\mu(a, \inv \circ \mu(\inv b, a)) = 
    \mu(a, \inv \circ \mu_{xy}(\iota_y' \circ \inv b, \iota_x' a)) =
    \mu_{xy}(\iota_x' a, \inv \circ \mu_{xy}(\inv \circ \iota_y'  b, \iota_x' a)) \le \iota_y' b
    $, where we used the left \propName of the \structName $P_{xy}$. But this means that 
    $\mu(a, \inv \circ \mu(\inv b, a)) \le b$. We conclude that 
    $P$ meets the left \propName condition. Applying the same procedure tells us that $P$ also has right \propName. Thus, $P$ is a \structName.
\end{proof}

Now that we know that \compStructName{s} are \structName{s}, we can compute quotients in these structures. We will now consider the solution of inequalities over languages as an application of \compStructName{s}.

\section{\CompStructName{s} and language inequalities}
\label{sec:sieved-langeq}

Language inequalities arise as the formalization of the problem of synthesizing
an unknown component in hardware and software systems. In this section, we
provide preliminaries on languages and discuss their properties and operations.
A fuller treatment of language properties can be found in~\cite{fsmeq1-iwls2004,
villa12}. Our objective is to show that commonly studied language structures are
\compStructName{s}, which allows us to axiomatically find their quotients per
the results of Section \ref{sec:sievedHeaps}.

\subsection{Operations on languages}
\label{subsec:eqlang:defn}

An alphabet is a finite set of symbols. The set of all finite strings over 
a fixed alphabet $X$ is denoted by $X^{\star}$. $X^{\star}$ includes the empty 
string $\epsilon$. A subset $L \subseteq X^{\star}$ is called a {\bf language} 
over alphabet $X$. \cite{ullman-automata2} is a standard reference on this subject.

A {\bf substitution} $f$ is a mapping of an alphabet $\Sigma$ to subsets of 
$\Delta^{\star}$ for some alphabet $\Delta$. The substitution $f$ is extended 
to strings by setting $f(\epsilon) = \{\epsilon\}$ and $f(xa) = f(x)f(a)$.
The following are well-studied language operations.
\begin{itemize}[leftmargin=*,noitemsep]
\item
Given a language $L$ over alphabet $X$ and an alphabet $V$,
consider the substitution $l\colon X \rightarrow 2^{(X \times V)^{\star}}$ defined as
$
l(x) = \setArg{(x,v) }{ v \in V }.
$
Then the language
$
L_{\uparrow V} = \setunion_{\alpha \in L} l(\alpha)
$
over alphabet $X \times V$ is the {\bf lifting} of language $L$ 
to alphabet $V$.
\item
Given a language $L$ over alphabet $X$ and an alphabet $V$,
consider the mapping $e\colon X \rightarrow 2^{(X \setunion V)^{\star}}$ defined as
$
e(x) = \setArg{\alpha x \beta }{ \alpha, \beta \in (V \setminus X)^{\star} }.
$
Then the language
$
L_{\Uparrow V} = \setunion_{\alpha \in L} e(\alpha)
$
over alphabet $X \setunion V$ is the {\bf expansion} of language $L$ to
alphabet $V$, i.e., words in 
$L_{\Uparrow V}$ are obtained from those in $L$ by inserting 
anywhere in them words from $(V \setminus X)^{\star}$.
Notice that $e$ is not a substitution and that 
$e(\epsilon) = \setArg{\alpha }{ \alpha \in V^{\star} }$.
\end{itemize}
The following proposition states that language liftings and expansions meet the
properties of concretization maps of a \compStructName. These results will be
used in the next section dealing with inequalities over languages.


\begin{proposition}
\label{prop:eqlang:defop:structure-preserving}
Liftings and expansions are order-preserving and commute with intersection and complementation.
\end{proposition}

\subsection{Composition of languages and inequalities involving languages}
\label{subsec:eqlang:comp}

Consider two systems $A$ and $B$ with associated languages $L(A)$ and $L(B)$.
The systems communicate with each other by a channel $U$ and with the
environment by channels $I$ and $O$.
The following two well-studied operators describe the external behavior
of the composition of $L(A)$ and $L(B)$.
\begin{definition}
Given the disjoint alphabets $I, U, O$, a language $L_1$ over $I \times U$, and a
language $L_2$ over $U \times O$, the {\bf synchronous composition} of
languages $L_1$ and $L_2$ is the language $(L_1)_{\uparrow O} \setint (L_2)_{\uparrow I}$,
denoted by $L_1 \bullet L_2$, defined over $I \times U \times O$.
\end{definition}
\begin{definition}
Given the disjoint alphabets $I, U, O$, a language $L_1$ over $I \setunion U$, and
a language $L_2$ over $U \setunion O$, the {\bf parallel composition} of
languages $L_1$ and $L_2$ is the language
$(L_1)_{\Uparrow O} \setint (L_2)_{\Uparrow I}$,
denoted by $L_1 \diamond L_2$, defined over $I \setunion U \setunion O$.
\end{definition}
\paragraph{Example. } Let $L_1 = \{a, aa\}$ be a language of the alphabet $\Sigma_1 = \{a, b\}$, and $\Sigma_2 = \{c, d\}$ be another alphabet for which $L_2 = \{c\}$ is a language. Then $L_1 \bullet L_2 = \{(a, c)\}$ and $L_1 \diamond L_2 = \{ac, ca, caa, aca, aac\}$.

Synchronous composition abstracts the parallel execution of modules in lock
step, assuming a global clock and instant communication by a broadcasting
mechanism, modeling the product semantics common in the hardware community. In
asynchronous composition modules execute independently at different speeds
assuming clocks which progress at arbitrary rates relative to one another,
modeling the interleaving semantics common in the software community. A
comparison can be found in~\cite{kurshan-asynch99}. Now we show that we can
interpret the above products as the source multiplication of a \compStructName.
For each product, we first need to identify a suitable indexing semilattice.
Then we need to build the appropriate \structName{s} and their maps.

\subsubsection{Synchronous equations}

\paragraph{Semilattice.} Suppose we have a disjoint family $F = \{\Sigma_i\}_{1 \le i \le n}$ of alphabets for some positive integer $n$, and let $S = 2^F$. Then $S$ is a semilattice under the operation of set union, i.e., if $x, y \in S$, we have $xy = x \setunion y$.

\paragraph{\StructName{s}.} For any $x \in S$, let $\card{x}$ be the cardinality of $x$. There exist natural numbers $k_1, \dots, k_{\card{x}}$ such that $x = \{\Sigma_{k_j}\}_{1 \le j \le \card{x}}  \subseteq F$ and $1 \le k_i < k_j \le n$ for $i < j$. We map each $x$ to a \structName as follows.
We define the alphabet over $x$ as $\alpha(x) = \Sigma_{k_1} \times \cdots \times \Sigma_{k_{\card{x}}}$, and we set $P_x = 2^{\alpha(x)^*}$. Source multiplication $\mu_x$ for $P_x$ is intersection, and involution $\inv_x$ is complementation. $(P_x, \le_x, \mu_x, \inv_x)$ is a Boolean lattice, thus a \structName, as shown in Section \ref{sec:preordHeaps}.

\paragraph{Concretizations.} For $x, y \in S$, $P_{xy}$ is clearly a \structName because $xy \in S$. 
We also define the \structName $P_{x,y} = 2^{\Sigma_{x,y}^*}$ for $\Sigma_{x,y} = \alpha(x) \times \alpha(y \setminus x)$ with source multiplication equal to set intersection and involution equal to complementation. Note that the only difference between $P_{xy}$ and $P_{x,y}$ is the order in which the alphabets $\Sigma_i$ appear in each: $P_{xy}$ contains all sets of finite strings over the alphabet $\alpha(xy)$, and $P_{x,y}$ contains all sets of finite strings over the alphabet $\alpha(x) \times \alpha(y - x)$. Thus, $P_{xy}$ and $P_{x,y}$ are isomorphic as sets. Let $\beta\colon P_{x,y} \to P_{xy}$ be this isomorphism, which is easily seen to be a \structName isomorphism. This allows us to define the concretization $\iota_x$ as follows:
{\scriptsize $
    \begin{tikzcd}
        & & P_{xy} \arrow[d, "\beta", leftarrow] 
         \\
        P_x \arrow[rr, "(\cdot) \uparrow_{\alpha(y - x)}", pos=0.7]
        \arrow[urr, dashed, "\iota_x"] & & P_{x,y} 
    \end{tikzcd}
$}
.

From Proposition
    \ref{prop:eqlang:defop:structure-preserving}, we know that $(\cdot)\uparrow_{\alpha(y - x)}$ is a
    \structName map. Thus, we have an $S$-\compStructName $\{(P_x, \le_x, \mu_x, \inv_x)\}_{x \in S}$.
Since \compStructName{s} are \structName{s} (Theorem \ref{kiybd9qybo}), for $A \in P_x$ and $B \in P_y$, an equation of the form $A \bullet z \le B$ has the largest solution $Z \in P_{xy}$ with
\[Z = \neg\left(\neg \beta' \left(B\uparrow_{\alpha(x \setminus y)}\right) \setint \,\,
                      \beta'' \left(A\uparrow_{\alpha(y \setminus x)}\right)\right),\]
where $\beta'\colon P_{y,x} \to P_{xy}$ and $\beta''\colon P_{x,y} \to P_{xy}$ are extensions of the alphabet permutations to languages, as described above.

\paragraph{Example.} Let $I$, $U$, and $O$ be disjoint alphabets. Then $S$ consists of all subsets of $\{I, O, U\}$. Let $i = \{I\}$, $u = \{U\}$, and $o= \{O\}$. The \structName $P_{iu}$ consists of all languages over the alphabet $I \times U$. $P_{uo}$ consists of all languages over $U \times O$. If $L_1 \in P_{iu}$, the concretization $\iota \colon P_{iu} \to p_{iuo}$ maps $L_1$ to a language over $I \times U \times O$. Observe that the order in which each alphabet appears is important and set from the beginning; this eliminates any potential ambiguities with the ordering of the alphabets (e.g., is it the alphabet $I \times U$ or $U \times I$?). By definition, this concretization map is $(\cdot)\uparrow_{O}$. In the same way, the concretization $\iota' \colon P_{uo} \to p_{iuo}$ is $\beta \circ (\cdot)\uparrow_{I}$, where $\beta: P_{uo,i} \to P_{iuo}$ permutes the symbols of the language so that they appear in the order $(a, b , c)$ with $a \in I$, $b \in U$, and $c \in O$. Thus, source multiplication is $\smult(L_1, L_2) = L_1\uparrow_{O} \setint \beta\left(L_2\uparrow_{I}\right)$, which is the synchronous product.

\subsubsection{Asynchronous equations}

Now we form a semilattice $S$ whose elements are abstract sets and whose operation
is set union. Let $x \in S$, and define $P_x = 2^{x^*}$. For $y \in S$, the concretization
$
    \begin{tikzcd}
        P_x \arrow[r, "\iota"]
        & P_{xy}
    \end{tikzcd}
$
is $\iota = (\cdot)\Uparrow_{y - x}$. Proposition \ref{prop:eqlang:defop:structure-preserving} shows that $\iota$ is a \structName map. Thus, we have a \compStructName $\{(P_x, \le_x, \mu_x, \inv_x)\}_{x \in S}$.

Since \compStructName{s} are \structName{s} (Theorem \ref{kiybd9qybo}), we are in a position to solve language equations under asynchronous composition. Let $x, y \in S$, $A \in P_x$ and $B \in P_y$. The largest solution to the equation $A \diamond z \le B$ yields $Z \in P_{xy}$ with $Z = \neg\left(\neg B \Uparrow_{x - y} \setint \,\,A\Uparrow_{y - x}\right)$.

\paragraph{Example.} As before, let $I$, $U$, and $O$ be disjoint alphabets, and let $I, U, O \in S$, where $S$ is a semilattice with the operation of set union. The \structName $P_{IU}$ consists of all languages over $I \setunion U$. Similarly, the \structName $P_{UO}$ consists of all languages over $U \setunion O$. The embedding $\iota \colon P_{IU} \to P_{IUO}$ is simply $(\cdot)\Uparrow_{O}$, and the embedding $\iota' \colon P_{UO} \to P_{IUO}$ is $(\cdot)\Uparrow_{I}$. Thus, for $L_1 \in P_{IU}$ and $L_2 \in P_{UO}$, source multiplication is $\mu(L_1, L_2) = L_1\Uparrow_O \cap L_2 \Uparrow_I$, which is the asynchronous product.

\section{Conclusions}
\label{sec:conclusions}

The comparison of the closed form computation of quotients ranging from language equations to AG contracts suggested     
a new algebraic structure, called {\em \structName}, endowed with the axioms
of preorders, together with a monotonic multiplication and an involution.
We showed that an \propName condition allows to solve equations over  
\structName{s}, and we gave the closed form of the solution. 
We showed that various theories qualify as \structName{s} and therefore admit
such explicit solution.
In particular, we showed that the conditions for being \structName{s} hold 
for Boolean lattices, assume-guarantee contracts,    
and for interface automata: in all cases we were able to derive 
axiomatically the quotients, which had been previously obtained by specific
analysis of each theory.
Finally we defined equations over \compStructName{s} to handle
components defined over multiple alphabets, and rederived as special cases the
solution of language equations known in the literature.


\subsection*{Acknowledgements}

We are grateful for the comments of our anonymous reviewers. This work was supported in part by NSF Contract CPS Medium 1739816; MIUR, Project ``Italian Outstanding Departments, 2018-2022''; INDAM, GNCS 2020, ``Strategic Reasoning and Automated Synthesis of Multi-Agent Systems''; University of Verona, Cooperint 2019, Program for Visiting Researchers.

\bibliographystyle{eptcs}
\bibliography{support/references}


\end{document}